\documentclass{article}
\usepackage{ijcai20}

\usepackage{times}

\usepackage{soul}
\usepackage{url}
\usepackage[hidelinks]{hyperref}
\usepackage[utf8]{inputenc}
\usepackage[small]{caption}
\usepackage{graphicx}
\usepackage{amsmath}
\usepackage{amsthm}
\usepackage{multirow}
\usepackage{algorithm,algorithmic}
\usepackage{NomaCommand}
\usepackage{booktabs}
\newtheorem{theorem}{Theorem}[section]
\newtheorem{lemma}{Lemma}[section]
\newtheorem{definition}{Definition}[section]
\newtheorem{problem}{Problem}

\newcommand{\argmin}{\mathop{\rm arg~min}\limits}

\title{RM-CVaR: Regularized Multiple $\beta$-CVaR Portfolio}

\author{
Kei Nakagawa$^1$\and
Shuhei Noma$^1$\and
Masaya Abe$^1$
\affiliations
$^1$Innovation Lab, Nomura Asset Management Co.,Ltd. ,Japan\\
\emails
\{kei.nak.0315, shuhei.mi.st, masaya.abe.428\}@gmail.com
}

\begin{document}

\maketitle

\begin{abstract}
The problem of finding the optimal portfolio for investors is called the portfolio optimization problem that mainly uses the expectation and variability of return (i.e., mean and variance).
Although the variance would be the most fundamental risk measure to be minimized, it has several drawbacks. 
Conditional Value-at-Risk (CVaR) is a relatively new risk measure that addresses some of the shortcomings of the well-known variance-related risk measures, and because of its computational efficiencies has gained popularity.
CVaR is defined as the expected value of the loss that occurs beyond a certain probability level ($\beta$).
However, portfolio optimization problems that use CVaR as a risk measure are formulated with a single $\beta$ and may output significantly different portfolios depending on how the $\beta$ is selected.
We confirm even small changes in $\beta$ can result in huge changes in the whole portfolio structure.
In order to improve this problem, we propose RM-CVaR: Regularized Multiple $\beta$-CVaR Portfolio.
We perform experiments on well-known benchmarks to evaluate the proposed portfolio.
Compared with various portfolios, RM-CVaR demonstrates a superior performance of having both higher risk-adjusted returns and lower maximum drawdown.
\end{abstract}

\section{Introduction}

The problem of finding the optimal portfolio for investors is called the portfolio optimization problem that mainly uses the expectation and variability of return (i.e., mean and variance\cite{markowitz1952portfolio}).
Although the variance would be the most fundamental risk measure to be minimized, it has several drawbacks. 
Controlling the variance does not only lead to low deviation from the expected return on the downside, but also on the upside.
Hence, quantile based risk measures have been suggested such as Value-at-Risk (VaR) that manage and control risk in terms of percentiles of the loss distribution. 
Instead of regarding both upside and downside of the expected return, VaR considers only the downside of the expected return as risk and represents the predicted maximum loss with a specified confidence level (e.g., 95\%).
VaR is incorporated into several regulatory requirements, like Basel Accord I\hspace{-.1em}I, and hence plays a particularly important role in risk analysis.
However, VaR, if studied in the framework of coherent risk measures \cite{artzner1999coherent}, lacks subadditivity, and therefore convexity,
in the case of general loss distributions (although it may be subadditive for special classes of them, e.g. for normal distributions). 
This drawback entails both inconsistencies with the well-accepted principle of diversification (diversification reduces risk).
For example, VaR of two different investment portfolios may be greater than the sum of the individual VaRs.
Also, VaR is nonconvex and nonsmooth and has multiple local minimum, while we seek the global minimum \cite{mcneil2005quantitative}. 
Besides, both variance and VaR ignores the magnitude of extreme or rare losses by their definition. 

Conditional Value-at-Risk (CVaR) that addresses these shortcomings of the variance and VaR is a relatively new risk measure and has gained popularity.
CVaR is defined as the expected value of the loss that occurs beyond a certain probability level ($\beta$).
\cite{pflug2000some} proved that CVaR is a coherent risk measure having subadditivity and convexity.
Additionally, \cite{rockafellar2000optimization} show that the minimization of CVaR results in a tractable optimization problem. 
For example, when the loss is defined as the minus return and a finite number of historical observations of returns are used in estimating CVaR, its minimization can be written as a
linear program and solved efficiently.

However, portfolio optimization problems that use CVaR as a risk measure are formulated with a single $\beta$ and may output significantly different portfolios depending on how the $\beta$ is selected.
We evaluate how the portfolio changes as the $\beta$ level changes with well-known benchmarks.
This is similar to the ''error maximization'' that \cite{michaud1989markowitz} points out in the case of the mean-variance portfolio.
\cite{ardia2017impact,nakagawa2018risk} empirically showed minimum variance portfolio weights are highly sensitive to the inputs.
We confirm even small changes in $\beta$ can result in huge changes in the whole portfolio structure.

On the other hand, many papers either ignore transaction costs or only subtract ad hoc transaction costs afterward \cite{shen2014doubly}.
When buying and selling assets on the markets, the investors incur in payment of commissions and other costs, globally defined transaction costs, that are charged by the brokers or the financial institutions playing the role of intermediary. 
Transaction costs represent the most important feature to account for when selecting a real portfolio, since they diminish net returns and reduce the amount of capital available for future investments \cite{mansini2015portfolio}.

In order to improve these problems, in this paper, we propose RM-CVaR: Regularized Multiple $\beta$-CVaR portfolio to bridge the gap between risk minimization and cost reduction. 
To control transaction cost, we impose $L$1-regularization term as in \cite{demiguel2009generalized,shen2014doubly}.
We prove that the RM-CVaR Portfolio optimization problem is written as a linear programming problem like the single $\beta$-CVaR portfolio.
We also perform experiments to evaluate the proposed portfolio.
Compared with various portfolios, the RM-CVaR portfolio demonstrates a superior performance of having both higher risk-adjusted returns and lower maximum drawdown.



\section{Related Work}
Because of these regularization and sparsity-inducing properties, there has been substantial
recent interest in $L$1-regularization in the statistics and optimization, beginning with \cite{tibshirani1996regression}.
Our approach is similar in spirit to \cite{zou2008regularized}, that estimates the simultaneously estimating multiple conditional quantiles.

The more conventional regularization models have been investigated for the Markowitz's portfolio optimization problem by \cite{demiguel2009generalized,fan2012vast,shen2014doubly}.
They show superior portfolio performances when various types of norm regularities are combined.
Analogously, \cite{gotoh2011role} consider $L$1 and $L$2-norms for the mean-CVaR problem.
Our paper extends this literature to multiple CVaR.

It is not hard to see that there is a connection between the portfolio optimization (risk minimization) and the optimization in machine learning \cite{gotoh2014interaction}.
Both estimate models that would achieve good out-of-sample performance. 

\cite{shen2015portfolio,shen2016portfolio} propose to employ the bandit learning framework to
attack portfolio problems.
\cite{shen2015portfolio} presented a bandit algorithm for conducting online portfolio choices by effectually exploiting correlations among multiple arms. 
\cite{shen2016portfolio} proposed an online algorithm that leverages Thompson sampling into the sequential decision-making process for portfolio blending. 
Also, \cite{shen2017portfolio,shen2019kelly} apply a subset resampling algorithm into the mean-variance portfolio and the Kelly growth optimal portfolio to obtain promising results. 
Through resampling subsets of the original large datasets, \cite{shen2017portfolio,shen2019kelly} constructed the associated subset portfolios with more accurately estimated parameters without requiring additional data.
However, these studies do not take transaction costs or turnover into account.

Interactions from portfolio optimization to machine learning include \cite{gotoh2005linear,takeda2008nu}.
\cite{gotoh2005linear} first have pointed out the common mathematical structure employed both in the class of machine learning methods known as $\nu$-support vector machines ($\nu$-SVMs) and in the
CVaR minimization. 
On the other hand, \cite{takeda2008nu} were the first to point out that the model of \cite{gotoh2005linear} is equivalent to the machine learning methods called E$\nu$-SVC.

Applying our method to machine learning algorithms is a major future task.

\section{Preliminary}
In this section, we define VaR and CVaR, and then formulate a portfolio optimization problem using them.

Let $r_i$ be the return of stock $i (1 \leq i \leq n)$ and $w_i$ be the portfolio weight for stock $i$.
Here, $r_i$ is a random variable and follows the continuous probability density function $p(r)$.

We denote $r = (r_1,..., r_n)^T $ and $w = (w_1,..., w_n)^T$.
Let $L (w, r)$ be loss function e.g. $L (w, r) = - w^T r$.
The probability that the loss function is less than $\alpha$ is 
\begin{equation}
    \Phi(w,\alpha) = \int_{L (w, r)\leq \alpha} p(r) dr
\end{equation}

When $w$ is fixed, $\Phi(w,\alpha)$ is non-decreasing as a function of $\alpha$ and is continuous from the right, but is generally not continuous from the left.
For simplicity, assume $\Phi(w,\alpha)$ is continuous function with respect to $\alpha$.
Then, VaR and CVaR are defined as follows (Figure \ref{VaR_CVaR}).

\begin{definition}
\begin{equation}
    \barapp{VaR}{w}{\beta}
    := \alpha_{\beta}(w)
    = \min(\alpha:\Phi(w,\alpha)>\beta) 
\end{equation}
\end{definition}

\begin{definition}
\begin{align}
    \barapp{CVaR}{w}{\beta}
    & := \barapp{\phi}{w}{\beta} \\ \nonumber
    & = (1-\beta)^{-1}\int_{L(w,r) \geq \alpha_{\beta}(w)}L(w,r)p(r)dr
\end{align}
\end{definition}

\begin{figure}
  \centering
  \includegraphics[width=0.4\textwidth]{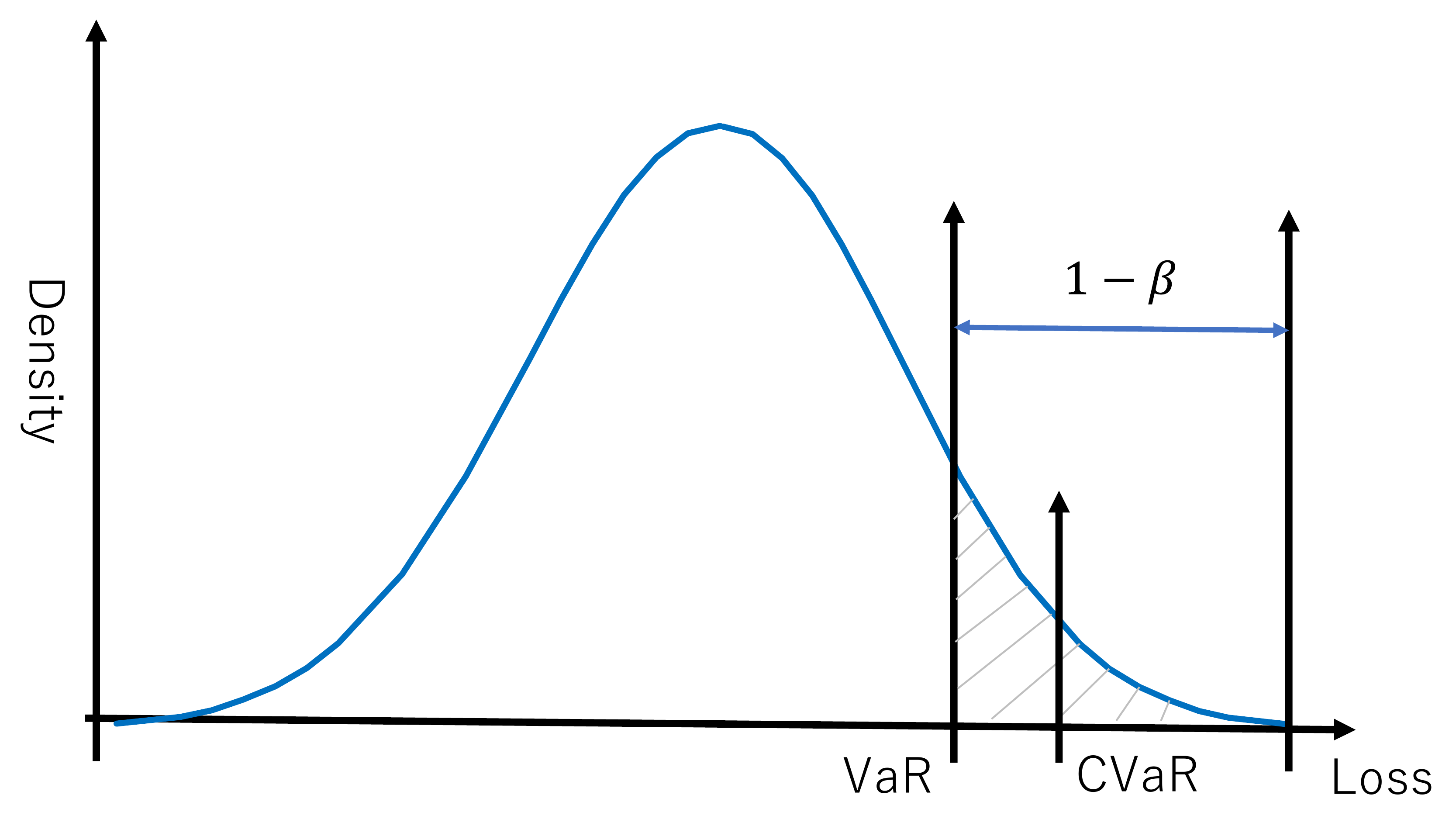}
  \caption{Illustration of VaR and CVaR.}
  \label{VaR_CVaR}
\end{figure}

It is difficult to directly optimize the above CVaR because the integration interval depends on VaR.
Therefore, to calculate $\phi_{\beta}(w)$, we also define $\barapp{F}{w,\alpha}{\beta}$ as
\begin{definition}

\begin{equation}
    \barapp{F}{w,\alpha}{\beta}
    := \alpha + (1-\beta)^{-1}\int_{R^n}[L(w,r)-\alpha]^+p(r)dr
    \end{equation}
    where $[t]^+ := \max(t,0)$.
\end{definition}

Then, the following relationship holds between
$\barapp{\phi}{w}{\beta}$ and
$\barapp{F}{w,\alpha}{\beta}$.
%
\begin{lemma}
\label{lemma1}
    For an arbitrarily fixed $w$,
    $\barapp{F}{w,\alpha}{\beta}$ is convex and continuously differentiable as a function of $\alpha$.
    $\barapp{\phi}{w}{\beta}$ is given by minimizing
    $\barapp{F}{w,\alpha}{\beta}$ with respect to $\alpha$.
    \begin{equation}
        \min_{\alpha} \barapp{F}{w,\alpha}{\beta}
        =
        \barapp{\phi}{w}{\beta}
    \end{equation}
    In this formula, the set consisting of the values of $\alpha$ for which the minimum is attained, namely
    \begin{equation}
        A_{\beta}
        =
        \argmin_{\alpha} \barapp{F}{w,\alpha}{\beta}
    \end{equation}
    is a nonempty closed bounded interval.
\end{lemma}
%
\begin{proof}
    Proof are given in \cite{rockafellar2000optimization}.
\end{proof}

CVaR is defined by the value of VaR, but it is possible to obtain the CVaR without obtaining VaR according to this Lemma.
If $X$ is a constraint that the portfolio must satisfy, the following Lemma holds for the formulation of a portfolio optimization problem using CVaR as the risk measure.
\begin{lemma}
\label{lemma2}
    Minimizing the CVaR over all $w \in X$ is equivalent to
    minimizing $\barapp{F}{w,\alpha}{\beta}$
    over all $(w,\alpha)\in X \times R$,
    in the sense that
    \begin{equation}
        \min_{w \in X} \barapp{\phi}{w}{\beta}
        =
        \min_{(w,\alpha)\in X \times R}
            \min_{\alpha} \barapp{F}{w,\alpha}{\beta}
        .
    \end{equation}
    Furthermore, if $f(w, r)$ is convex with respect to $w$,
    then $\barapp{F}{w,\alpha}{\beta}$ is convex
    with respect to $(w,\alpha)$,
    and $\barapp{\phi}{w}{\beta}$ is convex with respect to $w$.
    If $X$ is a convex set,
    the minimization problem of $\barapp{\phi}{w}{\beta}$ on $w \in X$ can be formulated
    as a convex programming problem.
\end{lemma}
\begin{proof}
    Proof are given in \cite{rockafellar2000optimization}.
\end{proof}

We approximate the function $\barapp{F}{w,\alpha}{\beta}$ by sampling a random variable $r$ from the density function $p(r)$.
When ${r[1], r[2], ..., r[q]}$ are obtained by sampling or simple historical data, the function $\barapp{F}{w,\alpha}{\beta}$ is approximated as follows.
\begin{equation}
\barapp{F}{w,\alpha}{\beta}=\alpha + (q(1-\beta))^{-1}\sum_{k=1}^q[-w^T r_k-\alpha]^+
\end{equation}

Finally, we formulate the portfolio optimization problem with CVaR as a linear programming problem as bellow.

\begin{align}
\min_{w,\alpha,u_1,...,u_q} & \alpha + (q(1-\beta))^{-1}\sum_{k=1}^q u_k \label{minCVaR} \\ 
s.t.  & u_k \geq -w^T r[k] -\alpha \quad (k=1,...,q)\\
& u_k \geq 0 \quad (k=1,...,q) \\
& 1^T w = 1 \\
& w_j \geq 0 \quad (j=1,...,n) \label{nonneg}
\end{align}

Here, $1^T w = 1$ indicates the sum of all the portfolio weights always equals one, and 1 (left side) denotes a column vector with ones. Beside, $w_j \geq 0$ indicates that investors take a long position of the $j$-th asset,
\section{RM-CVaR: Regularized Multiple $\beta$-CVaR Portfolio}
In this section, we propose a model that takes into account the multiple CVaR values.
The formulation is to minimize the margin between multiple $\beta$ levels of CVaR.

Let $C_{\beta_k}$ be the value of CVaR obtained by solving Eq. (\ref{minCVaR})-(\ref{nonneg}).
Then, minimizing $C$ considering $C_{\beta_k}$ is a main problem of this research.
\begin{problem}
\label{P1}
\begin{align}
    \min_{(w, C)\in X \times R} & C \\
    s.t.  & \barapp{\phi}{w}{\beta_k} \leq C + C_{\beta_k} \quad (k = 1, \ldots, K)
\end{align}
\end{problem}
Let $\barapp{F}{w,\alpha}{\beta}$ be the function likewise Lemma \ref{lemma1}
\begin{equation}
    \barapp{\phi}{w}{\beta}
    =
    \min_{\alpha} \barapp{F}{w, \alpha}{\beta}
    \label{PhiwBeta}
\end{equation}
Using Eq. (\ref{PhiwBeta}), Problem \ref{P1} can be written as follows.
\begin{problem}
\begin{align}
    \min_{(w,C)\in X \times R} & C \\
    s.t. &
        \min_{\alpha_k} \barapp{F}{w, \alpha_k}{\beta_k}
        \leq
        C + C_{\beta_k} \quad (k = 1, \ldots, K)
\end{align}
\label{mCVaR_v1}
\end{problem}
Let $\alpha=(\alpha_1, \cdots, \alpha_K)^T$ and we consider the following Problem \ref{mCVaR_v2}.
\begin{problem}

\begin{align}
    \min_{(w,C,\alpha)\in X \times R \times R^m} & C \\
    s.t. &
        \barapp{F}{w, \alpha_k}{\beta_k}
        \leq
        C + C_{\beta_k} \quad (k = 1, \ldots, K)
\end{align}
\label{mCVaR_v2}
\end{problem}

Here, the following Lemma holds between Problem \ref{mCVaR_v1} and \ref{mCVaR_v2}.
\begin{lemma}
    (1)If $(w^*,C^*)$ is the optimal value for Eq. (\ref{mCVaR_v1}),
    $(w^*,C^*,\alpha^*)$ is the optimal value of Eq. (\ref{mCVaR_v2}).
    (2)If $(w^{**},C^{**},\alpha^{**})$ is the optimal value for Eq. (\ref{mCVaR_v2}),
    $(w^{**},C^{**})$ is the optimal value for Eq. (\ref{mCVaR_v1}).
    \label{lemma3}
\end{lemma}
\begin{proof}
    Assume that $(w^*,C^*)$ is the optimal value for Problem \ref{mCVaR_v1}.
    Because $(w^*,C^*)$ is a feasible solution of Problem \ref{mCVaR_v1}, 
    $\min_{\alpha_i} \barapp{F}{w^*, \alpha_k}{\beta_k} \leq C^* + C_{\beta_k}$ holds.
    Define $\alpha^*=(\alpha_1, \ldots, \alpha_K)^T$ as
    $
        \alpha_k^*
        :=
        {\rm argmin}_{\alpha_k} \barapp{F}{w^*, \alpha_k}{\beta_k}.
    $
    Then, $(w^*, C^*, \alpha^*)$ is a feasible solution of Problem \ref{mCVaR_v2}
    since $\barapp{F}{w^*, \alpha_k^*}{\beta_k} \leq C^* + C_{\beta_k}$ holds.
    If $(w^*, C^*, \alpha^*)$ is not the optimal solution of Problem \ref{mCVaR_v2},
    there exists a feasible solution $(\hat{w}, \hat{C}, \hat{\alpha})$ satisfying $\hat{C} < C^*$.
    Then, $
        {\rm min}_{\alpha_k}
            \barapp{F}{
                \hat{w}, \hat{\alpha}_k}{\beta_k
            }
        \leq
        \hat{C} + C_{\beta_k}
        (k = 1, ..., K)
    $ holds.
    Therefore $(\hat{w}, \hat{C})$ is a feasible solution of Problem \ref{mCVaR_v1}, which contradicts that $C*$ is the optimal solution of Problem \ref{mCVaR_v1}.
    %
    Assume that $(w^{**},C^{**},\alpha^{**})$ is the optimal value for Problem \ref{mCVaR_v2}.
    Then, because $(w^{**},C^{**},\alpha^{**})$ is a feasible solution of Problem \ref{mCVaR_v2},
    $
        \barapp{F}{w^{**},\alpha_i^{**}}{\beta_i}
        \leq
        C^{**} + C_{\beta_i} \; (i=1,...,m)
    $ holds.
    $(w^{**}, C^{**})$ is a feasible solution for Problem \ref{mCVaR_v1} since
    $
        {\rm min}_{\alpha_i}
            \barapp{F}{w^{**},\alpha_i}{\beta_i}
            \geq
            \barapp{F}{w^{**},\alpha_i^{**}}{\beta_i}
            \leq
            C^{**} + C_{\beta_i}(i=1,...,m)
    $
    holds.
    if $(w^{**},C^{**})$ is not the optimal solution of
    Problem \ref{mCVaR_v1}, there exists a feasible solution $(\hat{w},\hat{C})$ satisfying $\hat{C} < C^{**}$.
    Define
    $
        \hat{\alpha}
        =
        (\hat{\alpha}_1,...,\hat{\alpha}_m)^T
    $ as $
        \hat{\alpha}_i := \argmin_{\alpha_i} F_{\beta_i}(\hat{w},\alpha_i)
    $.
    Then, $(\hat{w},\hat{C},\hat{\alpha})$ is a feasible solution of Problem \ref{mCVaR_v2}, which contradicts that $C^{**}$ is the optimal solution of Problem \ref{mCVaR_v2}.
\end{proof}
According to Lemma \ref{lemma3}, Problem \ref{P1} and \ref{mCVaR_v2} are a equivalent problem.
When ${r[1], \ldots, r[Q]}$ are obtained by sampling, 
the function $\barapp{F}{w, \alpha}{\beta}$ is approximated as follows.
\begin{equation}
    \barapp{F}{w, \alpha}{\beta}
    \simeq
    \alpha + \frac{1}{Q \paren{1 - \beta}}
    \sum_{q=1}^Q
    [ -w^\top r[q] - \alpha]^+
\end{equation}

Finally, we derive the Regularized Multiple $\beta$-CVaR Portfolio, where the objective is to minimize multiple CVaR values and control the portfolio turnover. 
The changes of the turnover during each rebalancing period are directly related to transaction costs, market impacts and taxes. 
Controling the portfolio turnover is realized through imposing
the $L1$-regularization term as
\begin{equation}
    \| w - w^- \|_1 = \sum_{i=1}^n |w_i - w_i^-|
\end{equation}
where $w_i^-$ denotes the portfolio weight before rebalancing.


From the above, the Regularized Multiple $\beta$-CVaR Portfolio optimization problem can be formulated as follows:
\begin{problem}
\begin{align}
        \min_{(w,C,\alpha) \in X \times R \times R^K}
        & C + \lambda \| w - w^- \|_1 \\
        {\rm s.t.}
        &
            \barapp{\tilde{F}}{w, \alpha_k}{\beta_k} \leq C + {C}_{\beta_k} (k = 1, \ldots, K)
    \end{align}
    \label{mCVaR_v3}
\end{problem}

We can easily proof Problem \label{mCVaR_v3} is linear programming problem similar to the usual CVaR minimization problem.
%
\begin{theorem}
    The Regularized Multiple $\beta$ CVaR Portfolio optimization problem
    is equivalent to the following linear programming problem.
    \begin{align*}
        \underset{
            C, w, \alpha, t, u
        }{
            {\rm min}
        }
        & \; C + \sum_{i = 1}^n u_i \\
        {\rm subject \; to}
        & \; u_i \geq \lambda \paren{w_i - w_i^-} \\
        & \; u_i \geq - \lambda \paren{w_i - w_i^-} \\
        & \; t_{qk} \geq 0 \\
        & \; t_{qk} \geq - w^\top r[q] - \alpha_k \\
        & \;
            \alpha_k + \frac{1}{Q \paren{1 - \beta_k}}
            \sum_{q=1}^Q t_{qk}
            \leq
            C + C_{\beta_k} \\
        & 1^T w = 1 \\
        & w_i \geq 0 \quad (i=1,...,n)
    \end{align*}
\label{thm}
\end{theorem}
\begin{proof}
    Using a standard approach in optimization, we replace each absolute value term
     $\lambda \| w - w^- \|_1$ with $softmax$.
    Then objective and constraints are all linear.
\end{proof}
\section{Experiment}
In this section, we will report the results of our empirical studies with well-known benchmarks.
First, we evaluate how the portfolio changes as the $\beta$ level changes.
Depending on how $\beta$ is chosen, a completely different portfolio may be constructed.
Next, we compare the out-of-sample performance among several portfolio strategy including our proposed.

\subsection{Dataset}
In the experiments, we use well-known academic benchmarks called Fama and French (FF) datasets \cite{fama1992cross} to ensure the reproducibility of the experiment.
This FF dataset is public and is readily available to anyone.
The FF datasets have been recognized as standard datasets and heavily adopted in finance research because of its extensive coverage to asset classes and very long historical data series.
We use FF25 dataset and FF48 dataset.
For example, FF25 dataset includes 25 portfolios formed on the basis of size and book-to-market ratio and FF48 dataset contains monthly returns of 48 portfolios representing different industrial sectors.
We use both datasets as monthly data from January 1989 to December 2018. 

\subsection{Experimental Settings}
\label{sec:Experimental_Settings}
In our empirical studies, the tested portfolio models have the following meanings:
\begin{itemize}
    \item ‘‘1/N’’ stands for equally-weighted (1/N) portfolio \cite{demiguel2007optimal}.
    \item ‘‘MV’’ stands for minimum-variance portfolio. We use the latest 10 years (120 months) to calculate covariance matrix.
    \item ‘‘DRP’’ stands for the doubly regularized miminum-variance portfolio \cite{shen2014doubly}. We use the latest 10 years (120 months) to calculate covariance matrix, and set combinations of two coefficients for regularization terms to  $\lambda_1$ = $\{0.001, 0.005, 0.01, 0.05\}$ and $\lambda_2$ = $\{0.001, 0.005, 0.01, 0.05\}$.
    \item ‘‘EGO’’ stands for the Kelly growth optimal portfolio with ensemble learning \cite{shen2019kelly}. We set $n_1$ (number of resamples) = 50, $n_2$ (size of each resample) = 5$\tau$, $\tau$ (number of periods of return data) = 120, $n_3$ (number of resampled subsets) = 50, $n_4$ (size of each subset) = $n^{0.7}$, $n$ is number of assets.
    \item ‘‘CVaR’’ stands for minimum CVaR portfolio with $\beta$. We implemnet 5 patterns of $\beta$ = \{0.95, 0.96, 0.97, 0.98, 0.99\}, and use the latest 10 years (120 months) to calulate each model.
    \item ‘‘ACVaR’’ stands for the average portfolio calculated by the average of minimum CVaR portfolio of different $\beta$ = \{0.95, 0.96, 0.97, 0.98, 0.99\} at each time point.
    \item ‘‘RM-CVaR’’ stands for our proposed portfolio. We set $K$ = 5 ($k=1, ... ,K$) as 5 patterns of $\beta_k$ = \{0.95, 0.96, 0.97, 0.98, 0.99\} to calculate $C_{\beta_k}$ and set $Q$ (number of sampling periods of return data) as \{10 years (120 months), 7 years (84 months)\}. For the coefficient of the regularization term, we implement 4 patterns of $\lambda$ = $\{0.001, 0.005, 0.01, 0.05\}$.
    We also implement $\lambda = 0 $ to compare with the best RM-CVaR. The RM-CVaR portfolio presented in Algorithm \ref{algo} is straightforward to implement.
\end{itemize}

We use the first-half period, from January 1989 to December 2003, as the in-sample period to decide the hyper-parameters of each model. 
After that, we use the second half-period, from January 2004 to December 2018 as the out-of-sample periods. 
Each portfolio is updated by sliding one-month-ahead and carrying out a monthly forecast. 
\begin{algorithm}
\caption{RM-CVaR Portfolio}
    \begin{algorithmic}[1]
        \renewcommand{\algorithmicrequire}{\textbf{Input:}}
        \renewcommand{\algorithmicensure}{\textbf{Output:}}
        \REQUIRE
            $K$ probability levels $\beta_k \in \paren{0, 1} \; \paren{k = 1, \ldots, K}$, \\
            a number of sampling periods $Q \in {\iset}^+$, \\
            a coeffient of the regularization term $\lambda \in \rset^+$ and \\
            a return matrix $Y \in \realmat{n}{\paren{T+Q}}$
        \ENSURE
            a weight matrix $W \in \realmat{n}{\paren{T+1}}$
        \FOR {$t = 1, \ldots, T+1$}
            \STATE
                $R \leftarrow Y \intset{t : Q+t-1}$
            \STATE
                Solve the linear programming introduced in \\
                Theorem \ref{thm}
            \STATE 
                Contain the solution $w^*$ to $W \intset{t}$
        \ENDFOR
        \RETURN $W$ 
    \end{algorithmic} 
    \label{algo}
\end{algorithm}
\subsection{Performance Measures}

In the first experiment, we define the weight difference of two minimum CVaR portfolios which have $\beta_i$ and $\beta_j$ are as bellow.
\begin{equation}
    {\bf Diff} = \frac{1}{T}\sum_{t=1}^{T} ||w_t^{\beta_i}-w_t^{\beta_j}||_1
\label{Diff}
\end{equation}
We set $\beta_i$ and $\beta_j$ as $\{0.96,0.95\}$, $\{0.97,0.96\}$, $\{0.98,0.97\}$, and $\{0.99,0.98\}$.
A large Diff indicates that the two portfolios are different.
This measure is similar to turnover defined below.

Next, we compare the out-of-sample performance of the portfolios.
In evaluating the portfolio strategy, we use the following measures that are widely used in the field of finance \cite{brandt2010portfolio}.

The portfolio return at time $t$ is defined as
\begin{equation}
    R_t = \sum_{i=1}^n r_{it}w_{it-1} 
\end{equation}
where $r_{it}$ is the return of $i$ asset at time $t$, $w_{it-1}$ is the weight of $i$ asset in the portfolio at time $t-1$, and $n$ is the number of asset. 
We evaluate the portfolio strategy by its annualized return (AR), risk as the standard deviation of return (RISK), risk/return (R/R) as return divided by risk as for the portfolio strategy. R/R is a risk-adjusted return measure for a portfolio strategy.  
\begin{align}
    {\bf AR} &= \prod_{t=1}^T (1+R_t)^{12/T}-1 \\
    {\bf RISK} &= \sqrt{\frac{12}{T-1}\times(R_t-\mu)^2}\\
    {\bf R/R} &= {\bf AR}/{\bf RISK}
\end{align}

Here, let $\mu = (1/T) \sum_{t=1}^T R_t$ be the average return of the portfolio.

We also evaluate maximum drawdown (MaxDD), which is yet another widely used risk measures \cite{magdon2004maximum,shen2017portfolio}, for the portfolio strategy:
Namely, MaxDD is defined as the largest drop from an extremum:
\begin{align}
    {\bf MaxDD} &= \min_{k \in [1,T]}\left(0,\frac{W_k}{\max_{j \in [1,k]} W_{j}}-1\right) \\
    W_k &= \prod_{l=1}^k (1+R_l).
\end{align}
where $W_k$ be the cumulative return of the portfolio until time $k$.

The turnover (TO) indicates the volumes of rebalancing. 
Since a high turnover inevitably generates high explicit and implicit trading costs thus reducing the portfolio return, it has been recognized as an important performance metric. 
The one-way annualized turnover is calculated as an average absolute value of the rebalancing trades over all the trading periods:
\begin{align}
    {\bf TO} = \frac{12}{2(T-1)}\sum_{t=1}^{T-1} ||w_t-w_{t-1}^{+}||_1
\end{align}
where $T-1$ indicates the total number of the rebalancing periods and $w_{t-1}^{+}$ is the re-normalized portfolio weight vector before rebalance.

\begin{equation}
w_{t-1}^{+} = \frac{w_{t-1} \otimes r_{t}}{w_{t-1}\top r_{t}}
\end{equation}

where $r_{t}$ is the return vector of the assets at time $t$, $w_{t-1}$ is the weight vector at time $t-1$ and the operator $\otimes$ denotes the Hadamard product. 


\begin{table}[t]
\caption{The weight difference of two minimum CVaR portfolios in the out-of-sample period.}
\label{CVaR_Diff_Wgt}
\begin{tabular}{|c|r|r|r|r|r|}
\hline
      & \multicolumn{1}{c|}{$\beta$96-95} & \multicolumn{1}{c|}{$\beta$97-96} & \multicolumn{1}{c|}{$\beta$98-97} & \multicolumn{1}{c|}{$\beta$99-98} & \multicolumn{1}{c|}{Avg} \\ \hline
FF25 & 42.66\%                     & 43.72\%                     & 48.77\%                     & 57.61\%                     & 48.19\%                  \\ \hline
FF48 & 23.95\%                     & 36.13\%                     & 25.20\%                     & 72.28\%                     & 39.39\%                  \\ \hline
Avg  & 33.31\%                     & 39.93\%                     & 36.99\%                     & 64.95\%                     & 43.79\%                  \\ \hline
\end{tabular}
\end{table}

\subsection{Experimental Results}
Table \ref{CVaR_Diff_Wgt} shows the weight difference of two minimum CVaR portfolios in the out-of-sample period.
Only 1\% difference in the $\beta$ level changes the portfolio weight on average by 48\% in FF25 and 39\% in FF48.
We confirm CVaR portfolio weights are highly sensitive to the $\beta$ levels.

\begin{figure*}[t]
 \begin{minipage}{0.5\hsize}
  \begin{center}
   \includegraphics[width=\textwidth]{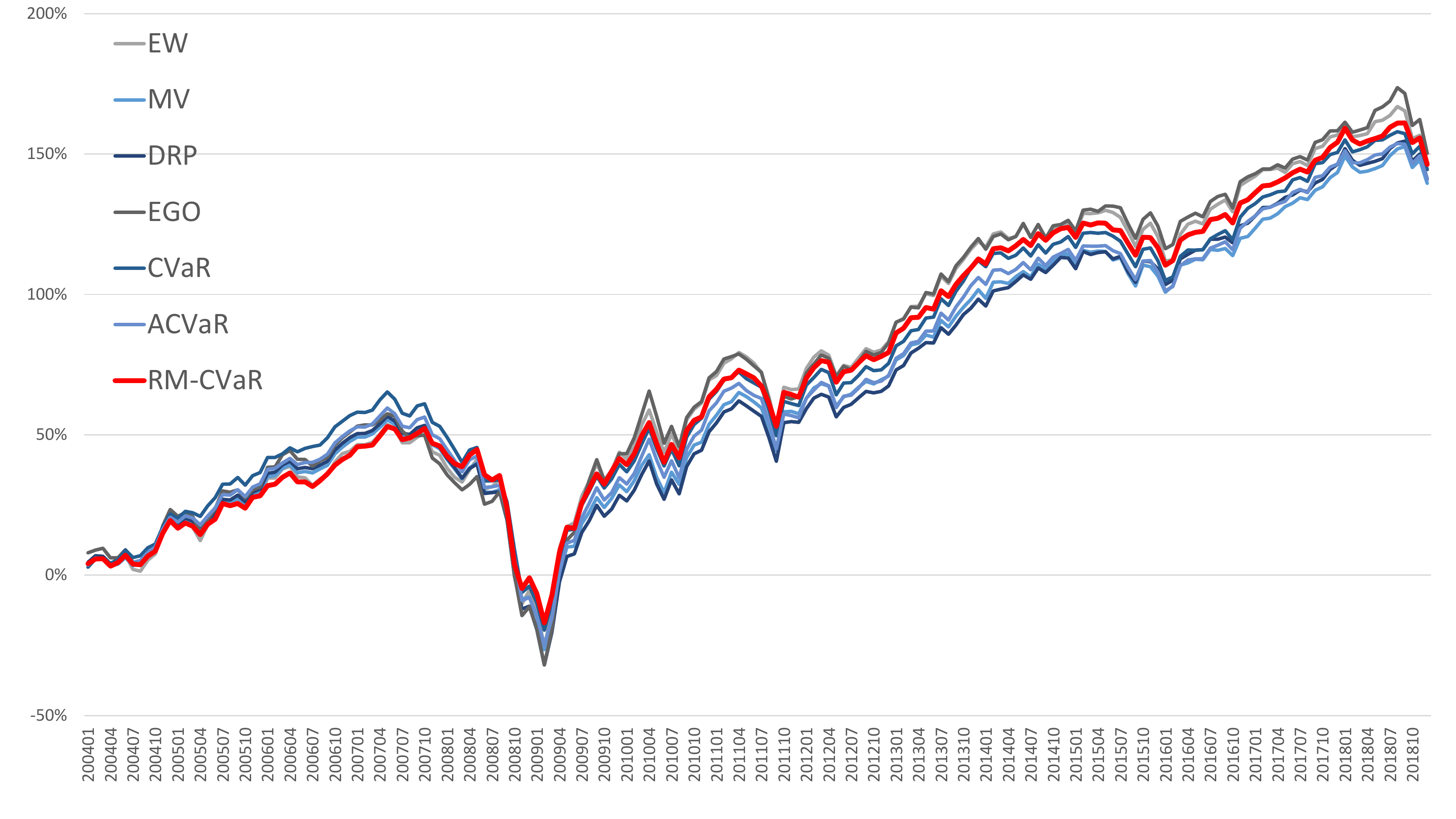}
  \end{center}
  \caption{The cumulative return in the out-of-sample period for FF25 datasets.}
  \label{fig:FF25}
 \end{minipage}
 \begin{minipage}{0.5\hsize}
  \begin{center}
   \includegraphics[width=\textwidth]{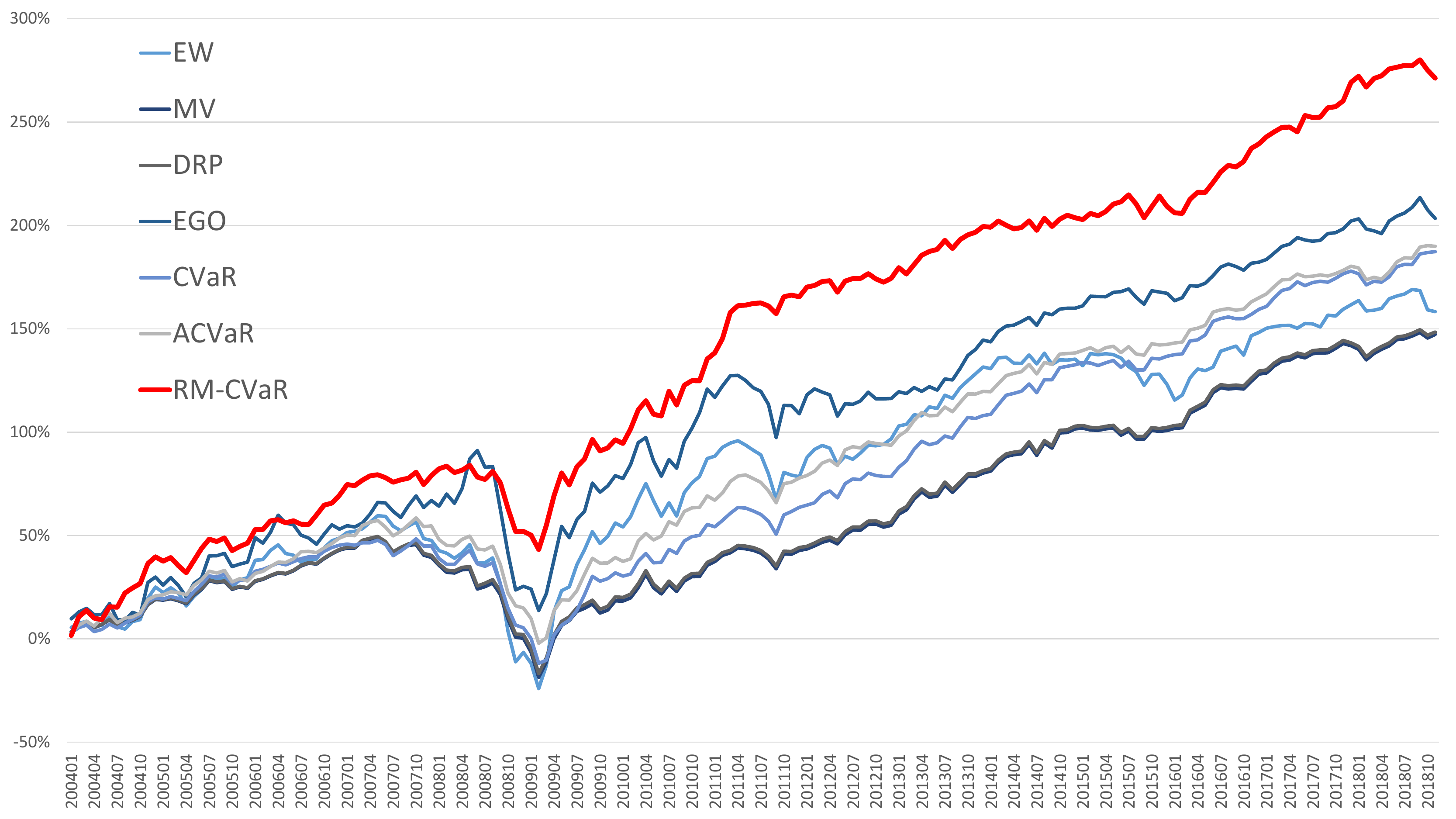}
  \end{center}
  \caption{The cumulative return in the out-of-sample period for FF48 datasets.}
  \label{fig:FF48}
 \end{minipage}
\end{figure*}

\begin{table*}[t]
\centering
\caption{The performance of each portfolio in out-of-sample period for FF25 dataset (upper panel) and FF48 dataset (lower panel).}
\label{results}
\small	
\begin{tabular}{|c|r|r|r|r|r|r|r|r|r|r|r|r|}
\hline
\multirow{2}{*}{FF25} & \multicolumn{1}{c|}{\multirow{2}{*}{EW}} & \multicolumn{1}{c|}{\multirow{2}{*}{MV}} & \multicolumn{1}{c|}{\multirow{2}{*}{DRP}} & \multicolumn{1}{c|}{\multirow{2}{*}{EGO}} & \multicolumn{1}{c|}{\multirow{2}{*}{ACVaR}} & \multicolumn{5}{c|}{CVaR}                                                                                                                 & \multicolumn{2}{c|}{RM-CVaR}                           \\ \cline{7-13} 
                      & \multicolumn{1}{c|}{}                    & \multicolumn{1}{c|}{}                    & \multicolumn{1}{c|}{}                     & \multicolumn{1}{c|}{}                     & \multicolumn{1}{c|}{}                           & \multicolumn{1}{c|}{95} & \multicolumn{1}{c|}{96} & \multicolumn{1}{c|}{97} & \multicolumn{1}{c|}{98} & \multicolumn{1}{c|}{99} & \multicolumn{1}{c|}{$\lambda=0$} & \multicolumn{1}{c|}{Best $\lambda$} \\ \hline
AR [\%]                   & 8.27                                   & 8.45                                   & 8.48                                    & 8.58                                    & 8.48                                          & 8.46                    & 8.36                    & 8.35                    & 8.73                    & 8.42                    & \textbf{9.03}          & 8.95                      \\ \hline
RISK [\%]                  & 18.13                                  & 15.24                                  & 15.67                                   & 18.71                                   & 15.62                                         & 16.15                   & 15.42                   & 15.54                   & 15.64                   & 15.74                   & 16.23                  & \textbf{15.11}            \\ \hline
R/R                   & 0.46                                     & 0.55                                     & 0.54                                      & 0.46                                      & 0.54                                            & 0.52                      & 0.54                      & 0.54                      & 0.56                      & 0.53                      & 0.56                     & \textbf{0.59}               \\ \hline
MaxDD [\%]                & -57.63                                 & -58.14                                 & -61.21                                  & -61.76                                  & -59.44                                        & -57.75                  & -56.75                  & -60.81                  & -59.54                  & -62.21                  & -54.14                 & \textbf{-52.81}           \\ \hline
TO [\%]                   & 16.95                                  & 31.10                                  & \textbf{8.75}                           & 71.52                                   & 22.19                                         & 17.46                   & 19.99                   & 24.97                   & 21.52                   & 29.98                   & 1000.98                & 33.57                     \\ \hline \hline

\multirow{2}{*}{FF48} & \multicolumn{1}{c|}{\multirow{2}{*}{EW}} & \multicolumn{1}{c|}{\multirow{2}{*}{MV}} & \multicolumn{1}{c|}{\multirow{2}{*}{DRP}} & \multicolumn{1}{c|}{\multirow{2}{*}{EGO}} & \multicolumn{1}{c|}{\multirow{2}{*}{ACVaR}} & \multicolumn{5}{c|}{CVaR}                                                                                                                 & \multicolumn{2}{c|}{RM-CVaR}                           \\ \cline{7-13} 
                      & \multicolumn{1}{c|}{}                    & \multicolumn{1}{c|}{}                    & \multicolumn{1}{c|}{}                     & \multicolumn{1}{c|}{}                     & \multicolumn{1}{c|}{}                           & \multicolumn{1}{c|}{95} & \multicolumn{1}{c|}{96} & \multicolumn{1}{c|}{97} & \multicolumn{1}{c|}{98} & \multicolumn{1}{c|}{99} & \multicolumn{1}{c|}{$\lambda=0$} & \multicolumn{1}{c|}{Best $\lambda$} \\ \hline
AR [\%]                   & 8.14                                   & 8.99                                   & 9.09                                    & 11.11                                   & 11.83                                         & 11.68                   & 10.79                   & 11.54                   & 11.96                   & 12.92                   & 15.75                  & \textbf{17.29}            \\ \hline
RISK [\%]                 & 19.27                                  & \textbf{11.77}                         & 11.77                                   & 20.61                                   & 12.53                                         & 12.27                   & 11.87                   & 12.42                   & 13.47                   & 14.49                   & 16.46                  & 15.61                     \\ \hline
R/R                   & 0.42                                     & 0.76                                     & 0.77                                      & 0.54                                      & 0.94                                            & 0.95                      & 0.91                      & 0.93                      & 0.89                      & 0.89                      & 0.96                     & \textbf{1.11}               \\ \hline
MaxDD [\%]                & -59.81                                 & -50.84                                 & -50.25                                  & -57.39                                  & -47.22                                        & -46.98                  & -45.21                  & -45.36                  & -48.23                  & -50.38                  & -35.29                 & \textbf{-34.93}           \\ \hline
TO [\%]                   & 36.73                                  & 27.48                                  & \textbf{17.15}                          & 75.80                                   & 37.04                                         & 41.31                   & 38.87                   & 35.37                   & 41.57                   & 38.36                   & 960.03                 & 750.48                    \\ \hline
\end{tabular}
\end{table*}

Table \ref{results} reports the overall performance measures of RM-CVaR, our proposed portfolio, and the compared 10 portfolios introduced in Section \ref{sec:Experimental_Settings}.
Among the comparisons of various portfolios, where the best performance is highlighted in bold.
In both datasets, the proposed RM-CVaR with $\lambda$ achieves the highest R/R and the lowest MaxDD.
Not surprisingly, RM-CVaR differs from ACVaR, which is the simple average of five CVaR portfolios.
RM-CVaR also exceeds individual $\beta$ levels of CVaR by R/R and MaxDD.
In FF25 datasets, RM-CVaR without $\lambda$ outperforms all the compared portfolios in terms of AR but has the worst TO.
Introducing the regularization term $\lambda$, the TO is considerably suppressed, and RISK, R/R and maxDD are also the best.
In FF48 datasets, In FF48, RM-CVaR with $\lambda$ has the best AR, R/R, and MaxDD, but the TO is very high. 
This is because the $\lambda$ selected in this experiment does not sufficiently suppress TO.
%

Furthermore, in order to compare the trend and dynamics of the each portfolio return, Figure \ref{fig:FF25} and \ref{fig:FF48} show the cumulative return over the out-of-sample periods for the FF25 and FF48 datasets. 
Although there is not much difference in the FF25 dataset, RM-CVaR apparently outperforms the others with the visible margins in the FF48 datasets in most of the time periods.
We can confirm that RM-CVaR avoides a large drawdown.

\section{Conclusion}
Our study makes the following contributions:
\begin{itemize}
    \item We propose RM-CVaR: Regularized Multiple $\beta$-CVaR Portfolio and prove that the optimization problem is written as a linear programming. 
    \item We demonstrate that the CVaR portfolio dramatically changes depending on the $\beta$ level.
    \item RM-CVaR is a superior performance of having both higher risk-adjusted returns and lower maximum drawdown.
\end{itemize}

Our future work includes incorporating the subsampling method such as \cite{shen2017portfolio,shen2019kelly}.


\end{document}